\documentclass[11pt]{article}
\usepackage[margin=1in]{geometry}
\usepackage{amsmath,amssymb,amsthm,mathtools}
\usepackage{bm}
\usepackage{booktabs}
\usepackage{enumitem}
\usepackage{microtype}
\usepackage[hidelinks]{hyperref}
\usepackage[nameinlink,noabbrev]{cleveref}

\setlist[itemize]{leftmargin=1.6em,itemsep=0.25em,topsep=0.25em}

\newtheorem{theorem}{Theorem}[section]
\newtheorem{proposition}[theorem]{Proposition}

\newtheorem{corollary}[theorem]{Corollary}
\theoremstyle{remark}

\newtheorem{example}[theorem]{Example}

\newcommand{\R}{\mathbb{R}}
\newcommand{\E}{\mathbb{E}}
\newcommand{\Q}{\mathbb{Q}}

\title{Gaussian Normalized Coordinates and Risk-Neutral CDF Deformations}
\author{Jian Sun}
\date{September 2026}

\begin{document}
\maketitle

\begin{abstract}
Normalized implied-volatility coordinates reveal no-arbitrage structure that is less transparent in strike space. This paper centers on a single quantity that gives those coordinates a direct probabilistic interpretation. If $h(z)$ denotes the risk-neutral CDF expressed in an increasing Gaussian normalized coordinate $z$, define the Gaussian CDF deformation
\begin{equation*}
  \eta(z)=\frac{h(z)-\Phi(z)}{\phi(z)}.
\end{equation*}
Thus
\begin{equation*}
  h(z)=\Phi(z)+\phi(z)\eta(z),
\end{equation*}
and $\eta$ measures the pointwise departure of the option-implied CDF from its Gaussian benchmark, normalized by the Gaussian density. Its differential transform
\begin{equation*}
  m(z)=1+\eta'(z)-z\eta(z)
\end{equation*}
then satisfies
\begin{equation*}
  h'(z)=\phi(z)m(z),
\end{equation*}
so $m$ is the option-implied density in normalized coordinates relative to the standard normal density. No-arbitrage becomes $m\ge0$, together with explicit Mills-ratio bounds on $\eta$.

The Bachelier and Black conventions give two concrete realizations of this same deformation. In Bachelier coordinates, $\eta$ is exactly the strike derivative of total normal implied standard deviation. If $q=1/s$ is reciprocal implied scale in the normalized coordinate $y=(K-F)/s(K)$, butterfly convexity is equivalent to the linear inequality
\begin{equation*}
  q''(y)+yq'(y)-q(y)\le0.
\end{equation*}
In Black coordinates, with the Fukasawa coordinates $a=k/v-v/2$ and $b=k/v+v/2$, the same CDF deformation is the slope $v'(k)$ of total implied volatility. The corresponding Gaussian density deformation reproduces the familiar Black butterfly factor. We also derive the associated equality families, one-coordinate reconstruction, and local-volatility formulas. The paper therefore interprets implied-volatility skew as a normalized deformation of the risk-neutral distribution rather than merely as a slope parameter.
\end{abstract}



\section{Introduction}
\label{sec:intro}

Implied volatility is quoted as a function of strike, but many of its structural properties become simpler after strike is normalized by the implied scale itself. In the Black convention, Fukasawa~\cite{Fukasawa2012} introduced the normalizing transformations
\begin{equation*}
  f_1(k)=\frac{k}{v(k)}-\frac{v(k)}{2},
  \qquad
  f_2(k)=\frac{k}{v(k)}+\frac{v(k)}{2},
\end{equation*}
where $k=\log(K/F)$ and $v(k)$ is total Black implied volatility. These are the negatives of the usual Black quantities $d_1$ and $d_2$. Their monotonicity under no-arbitrage, together with related slope bounds and one-coordinate parameterizations, has been studied by Fukasawa~\cite{Fukasawa2012}, Lucic~\cite{Lucic2021,Lucic2023}, and others. A finite-strike order-theoretic treatment, including the Bachelier analogue, is developed in the companion preprint Sun~\cite{Sun2026}.

The present paper studies the second-order structure of these normalized coordinates. The central object is not a differential operator by itself, but a function that measures how the option-implied distribution differs from its Gaussian benchmark. Let $z$ be an increasing normalized coordinate and let $h(z)$ denote the risk-neutral CDF expressed in that coordinate. We define
\begin{equation*}
  \eta(z)=\frac{h(z)-\Phi(z)}{\phi(z)}.
\end{equation*}
Equivalently,
\begin{equation*}
  h(z)=\Phi(z)+\phi(z)\eta(z).
\end{equation*}
Thus $\eta$ is a Gaussian-normalized CDF deformation. The Gaussian benchmark itself is characterized by $\eta\equiv0$.

Differentiation produces a second deformation,
\begin{equation*}
  m(z)=1+\eta'(z)-z\eta(z),
\end{equation*}
for which
\begin{equation*}
  h'(z)=\phi(z)m(z).
\end{equation*}
Hence $m$ is the density in the normalized coordinate relative to the standard normal density. The two functions have complementary interpretations:
\begin{equation*}
  \eta
  \quad\text{measures CDF deformation,}\qquad
  m
  \quad\text{measures density deformation.}
\end{equation*}
Static no-arbitrage requires $m\ge0$, while the digital bounds $0\le h\le1$ give exact Mills-ratio barriers for $\eta$. Moreover, the inequality $m\ge0$ can be solved completely: admissible $\eta$ are in one-to-one correspondence with nondecreasing CDFs, or equivalently with nonnegative Gaussian-relative densities of total mass one.

The significance for implied volatility is that $\eta$ becomes an observable slope in both classical Gaussian pricing conventions. In Bachelier coordinates, let $s(K)$ be total normal implied standard deviation and set
\begin{equation*}
  y(K)=\frac{K-F}{s(K)}.
\end{equation*}
Then
\begin{equation*}
  \eta(y)=s'(K(y)).
\end{equation*}
Thus Bachelier implied-volatility skew is exactly the normalized difference between the risk-neutral CDF and the standard normal CDF. Writing the reciprocal implied scale as $q(y)=1/s(y)$ further turns density positivity into the linear differential inequality
\begin{equation*}
  q''(y)+yq'(y)-q(y)\le0.
\end{equation*}

In Black coordinates, let $v=v(k)$ denote total Black implied volatility as a function of log-forward moneyness $k=\log(K/F)$ and define
\begin{equation*}
  a(k)=\frac{k}{v(k)}-\frac{v(k)}{2},
  \qquad
  b(k)=\frac{k}{v(k)}+\frac{v(k)}{2}.
\end{equation*}
The corresponding CDF deformations are again given by implied-volatility slope:
\begin{equation*}
  \eta=v'(k)=\frac{dv(k)}{dk}.
\end{equation*}
The pricing-measure CDF in the $b$ coordinate and the share-measure CDF in the $a$ coordinate are therefore two Gaussian descriptions of the same skew function. Their density deformation reproduces the familiar Black butterfly factor.

The paper is organized around this interpretation of $\eta$. Section~\ref{sec:gaussian} develops the CDF deformation, its exact bounds, its density transform, and the complete solution of the no-arbitrage inequality. Section~\ref{sec:bachelier} shows that $\eta$ is Bachelier implied-volatility slope and derives the reciprocal-scale linearization. Section~\ref{sec:black} gives the parallel Black interpretation in the two Fukasawa coordinates. Section~\ref{sec:reconstruction} treats $\eta$ as the primitive input of a one-coordinate Black reconstruction. Section~\ref{sec:surface} shows how the same density deformation enters local-volatility formulas. Section~\ref{sec:dictionary} summarizes the Black--Bachelier correspondence.

\section{Gaussian CDF deformation and density deformation}
\label{sec:gaussian}

We write $\phi$ and $\Phi$ for the standard normal density and distribution function. Since $\phi(-z)=\phi(z)$, define the Mills ratio by
\begin{equation}
\label{eq:mills}
  R(z)=\frac{\Phi(-z)}{\phi(-z)}
  =\frac{\Phi(-z)}{\phi(z)}.
\end{equation}

Let $h(z)$ be an option-implied CDF expressed in an increasing normalized coordinate $z$. The central definition of the paper is
\begin{equation}
\label{eq:eta-def}
  \eta(z)=\frac{h(z)-\Phi(z)}{\phi(z)}.
\end{equation}
Equivalently,
\begin{equation}
\label{eq:gaussian-cdf-form}
  h(z)=\Phi(z)+\phi(z)\eta(z).
\end{equation}
The sign of $\eta$ has an immediate interpretation: $\eta(z)>0$ means that the risk-neutral CDF has accumulated more probability to the left of $z$ than the Gaussian benchmark, while $\eta(z)<0$ means that it has accumulated less.

\begin{theorem}[Gaussian CDF and density deformations]
\label{thm:gaussian-skew}
Assume that $h$ and $\eta$ are differentiable and related by \eqref{eq:gaussian-cdf-form}. Define
\begin{equation}
\label{eq:M-operator}
  m(z)=\mathcal M_z(\eta)
  :=1+\eta'(z)-z\eta(z).
\end{equation}
Then
\begin{equation}
\label{eq:H-z}
  h'(z)=\phi(z)m(z).
\end{equation}
Consequently,
\begin{equation*}
  h'(z)\ge0
  \quad\Longleftrightarrow\quad
  m(z)\ge0.
\end{equation*}
Moreover,
\begin{equation}
\label{eq:density-ratio}
  m(z)=\frac{h'(z)}{\phi(z)},
\end{equation}
so $m$ is the option-implied density in the normalized coordinate divided by the standard normal density.
\end{theorem}

\begin{proof}
Differentiate \eqref{eq:gaussian-cdf-form} and use $\phi'(z)=-z\phi(z)$:
\begin{equation*}
  h'(z)
  =\phi(z)+\phi'(z)\eta(z)+\phi(z)\eta'(z)
  =\phi(z)\bigl(1+\eta'(z)-z\eta(z)\bigr).
\end{equation*}
Since $\phi(z)>0$, the sign equivalence and \eqref{eq:density-ratio} follow.
\end{proof}

The Gaussian benchmark is therefore the distinguished origin of the deformation representation:
\begin{equation*}
  h=\Phi
  \quad\Longleftrightarrow\quad
  \eta=0,
  \qquad
  m=1.
\end{equation*}
Thus $\eta$ measures cumulative departure from Gaussianity, while $m-1$ measures local density departure from Gaussianity.

\begin{proposition}[Complete characterization of the Gaussian no-arbitrage inequality]
\label{prop:complete-M}
Let $\eta$ be differentiable and define $h$ by \eqref{eq:gaussian-cdf-form}. Then
\begin{equation*}
  \mathcal M_z(\eta)\ge0
  \quad\Longleftrightarrow\quad
  h \text{ is nondecreasing}.
\end{equation*}
If, in addition,
\begin{equation*}
  h(-\infty)=0,
  \qquad
  h(+\infty)=1,
\end{equation*}
then $m=\mathcal M_z(\eta)$ satisfies
\begin{equation}
\label{eq:m-normalization}
  m(z)\ge0,
  \qquad
  \int_{-\infty}^{\infty}m(u)\phi(u)\,du=1.
\end{equation}
Conversely, any locally integrable $m$ satisfying \eqref{eq:m-normalization} defines a CDF and a deformation by
\begin{equation}
\label{eq:h-from-m}
  h(z)=\int_{-\infty}^{z}m(u)\phi(u)\,du,
\end{equation}
\begin{equation}
\label{eq:eta-from-m}
  \eta(z)
  =\frac{\displaystyle\int_{-\infty}^{z}\phi(u)\bigl(m(u)-1\bigr)\,du}{\phi(z)},
\end{equation}
and this deformation satisfies $\mathcal M_z(\eta)=m$ almost everywhere.
\end{proposition}

\begin{proof}
The first equivalence follows from Theorem~\ref{thm:gaussian-skew}. Under the CDF boundary conditions, integrating $h'(z)=\phi(z)m(z)$ over $\mathbb R$ gives \eqref{eq:m-normalization}. Conversely, \eqref{eq:h-from-m} defines a nondecreasing function with limits zero and one. Since
\begin{equation*}
  \Phi(z)=\int_{-\infty}^{z}\phi(u)\,du,
\end{equation*}
subtracting $\Phi(z)$ from \eqref{eq:h-from-m} and dividing by $\phi(z)$ gives \eqref{eq:eta-from-m}. Theorem~\ref{thm:gaussian-skew} then gives $\mathcal M_z(\eta)=m$ almost everywhere.
\end{proof}

Proposition~\ref{prop:complete-M} shows that the inequality $\mathcal M_z(\eta)\ge0$ is completely solvable. An admissible $\eta$ is equivalent to a nonnegative Gaussian-relative density $m$ of unit Gaussian-weighted mass. The function $\eta$ is the cumulative effect of the local density distortion $m-1$.

\begin{proposition}[Mills-ratio barriers and zero-density curves]
\label{prop:mills-zero-density}
The digital bounds
\begin{equation*}
  0\le h(z)\le1
\end{equation*}
are equivalent to
\begin{equation}
\label{eq:mills-bounds}
  -R(-z)\le\eta(z)\le R(z).
\end{equation}
Moreover, the solutions of
\begin{equation*}
  \mathcal M_z(\eta)=0
\end{equation*}
are
\begin{equation}
\label{eq:M-equality}
  \eta_c(z)
  =\frac{c-\Phi(z)}{\phi(z)}
  =cR(z)-(1-c)R(-z),
\end{equation}
where $c$ is constant. Along such a curve, $h(z)\equiv c$.
\end{proposition}

\begin{proof}
Solving
\begin{equation*}
  0\le\Phi(z)+\phi(z)\eta(z)\le1
\end{equation*}
for $\eta$ gives \eqref{eq:mills-bounds}. If $\mathcal M_z(\eta)=0$, then Theorem~\ref{thm:gaussian-skew} gives $h'(z)=0$, so $h(z)=c$. Solving \eqref{eq:gaussian-cdf-form} for $\eta$ gives \eqref{eq:M-equality}.
\end{proof}

The equality curves have a direct financial meaning. When $z$ is increasing in strike, $m=0$ on an interval means that the corresponding interval carries zero risk-neutral density. Thus the Mills barriers control the admissible location of the CDF deformation, while the curves \eqref{eq:M-equality} describe the local zero-density boundary.

\section{Bachelier coordinates: implied-volatility slope as CDF deformation}
\label{sec:bachelier}

The general deformation $\eta$ acquires an especially transparent meaning in the Bachelier convention: it is exactly the strike derivative of total normal implied standard deviation. Thus Bachelier implied-volatility skew is not merely a slope parameter; it is the Gaussian-normalized displacement of the risk-neutral CDF.

We work first at a fixed maturity. Let $F\in\R$ be the forward and let $C(K)$ denote the undiscounted call price. The Bachelier total normal implied standard deviation $s(K)>0$ is defined as the unique value satisfying
\begin{equation}
\label{eq:bachelier-implied-s}
  C(K)
  =(F-K)\Phi\left(\frac{F-K}{s(K)}\right)
  +s(K)\phi\left(\frac{F-K}{s(K)}\right).
\end{equation}
Define the increasing normalized coordinate
\begin{equation}
\label{eq:bachelier-y}
  y(K)=\frac{K-F}{s(K)}.
\end{equation}
The finite-strike monotonicity of $y(K)$ under static no-arbitrage is established in Sun~\cite{Sun2026}. We therefore assume that $y(K)$ is strictly increasing, write $K=K(y)$ for its inverse, and define
\begin{equation}
\label{eq:s-of-y}
  s(y):=s(K(y)).
\end{equation}

Define the unit-scale Bachelier call and put kernels
\begin{equation}
\label{eq:bachelier-kernels}
  c(y)=\phi(y)-y\Phi(-y),
  \qquad
  p(y)=\phi(y)+y\Phi(y).
\end{equation}
Then
\begin{equation*}
  C(K)=s(K)c(y(K)),
  \qquad
  P(K)=s(K)p(y(K)).
\end{equation*}

\subsection{From the Bachelier price to the CDF deformation}

\begin{proposition}[Bachelier CDF deformation]
\label{prop:bachelier-cdf}
Assume that $s(K)$ and $C(K)$ are differentiable. Define the risk-neutral CDF in the normalized coordinate by
\begin{equation}
\label{eq:bachelier-h-def}
  h(y)=\Q\bigl(S_T\le K(y)\bigr),
\end{equation}
and define
\begin{equation}
\label{eq:bachelier-eta-def}
  \eta(y)=s'(K(y)).
\end{equation}
Then
\begin{equation}
\label{eq:bachelier-C-prime}
  C'(K)=s'(K)\phi(y)-\Phi(-y),
\end{equation}
and
\begin{equation}
\label{eq:bachelier-cdf}
  h(y)=\Phi(y)+\phi(y)\eta(y).
\end{equation}
Equivalently,
\begin{equation}
\label{eq:bachelier-eta-cdf}
  \eta(y)
  =\frac{h(y)-\Phi(y)}{\phi(y)}
  =s'(K(y)).
\end{equation}
\end{proposition}

\begin{proof}
Since $K-F=s(K)y(K)$, differentiation with respect to $K$ gives
\begin{equation*}
  1=s'(K)y+s(K)y'(K).
\end{equation*}
Also,
\begin{equation*}
  c'(y)=-\Phi(-y).
\end{equation*}
Differentiating $C(K)=s(K)c(y(K))$ therefore gives
\begin{align*}
  C'(K)
  &=s'(K)c(y)+s(K)c'(y)y'(K)\\
  &=s'(K)c(y)-\Phi(-y)\bigl(1-y s'(K)\bigr)\\
  &=s'(K)\phi(y)-\Phi(-y),
\end{align*}
which proves \eqref{eq:bachelier-C-prime}.

On the other hand, differentiability of the call-price curve gives
\begin{equation*}
  C'(K)=-\Q(S_T>K).
\end{equation*}
Hence
\begin{equation*}
  \Q(S_T\le K)=1+C'(K).
\end{equation*}
Evaluating at $K=K(y)$ and substituting \eqref{eq:bachelier-C-prime},
\begin{align*}
  h(y)
  &=1+C'(K(y))\\
  &=1-\Phi(-y)+\phi(y)s'(K(y))\\
  &=\Phi(y)+\phi(y)\eta(y).
\end{align*}
This proves \eqref{eq:bachelier-cdf} and \eqref{eq:bachelier-eta-cdf}.
\end{proof}

The proposition gives the main interpretation of the Bachelier skew:
\begin{equation*}
  s'(K(y))
  =\frac{\text{risk-neutral CDF}-\text{Gaussian CDF}}{\text{Gaussian density}}.
\end{equation*}
Thus the sign and magnitude of the Bachelier implied-volatility slope directly record how the risk-neutral distribution departs from the Gaussian benchmark in normalized coordinates.

\subsection{Reciprocal implied scale and density deformation}

Set
\begin{equation}
\label{eq:q-def}
  q(y)=\frac{1}{s(y)}.
\end{equation}
Then
\begin{equation}
\label{eq:K-y-q}
  K(y)=F+\frac{y}{q(y)},
\end{equation}
and
\begin{equation}
\label{eq:K-y-derivative}
  K'(y)=\frac{q(y)-yq'(y)}{q(y)^2}.
\end{equation}
The increasing-coordinate condition is therefore
\begin{equation}
\label{eq:bachelier-orientation}
  q(y)-yq'(y)>0.
\end{equation}

\begin{proposition}[Bachelier reciprocal-scale representation]
\label{prop:bachelier-M}
Under \eqref{eq:bachelier-orientation},
\begin{equation}
\label{eq:bachelier-eta-q}
  \eta(y)
  =-\frac{q'(y)}{q(y)-yq'(y)},
\end{equation}
and the Gaussian-relative density deformation is
\begin{equation}
\label{eq:bachelier-M-L}
  m(y)
  =\mathcal M_y(\eta)
  =-\frac{q(y)\bigl(q''(y)+yq'(y)-q(y)\bigr)}
  {\bigl(q(y)-yq'(y)\bigr)^2}.
\end{equation}
Moreover,
\begin{equation}
\label{eq:bachelier-density-q}
  C''(K)
  =-\frac{\phi(y)q(y)^3}
  {\bigl(q(y)-yq'(y)\bigr)^3}
  \bigl(q''(y)+yq'(y)-q(y)\bigr).
\end{equation}
\end{proposition}

\begin{proof}
Since $s(y)=1/q(y)$,
\begin{equation*}
  \frac{ds}{dy}=-\frac{q'(y)}{q(y)^2}.
\end{equation*}
Using \eqref{eq:K-y-derivative},
\begin{equation*}
  s'(K(y))
  =\frac{ds/dy}{dK/dy}
  =-\frac{q'(y)}{q(y)-yq'(y)},
\end{equation*}
which proves \eqref{eq:bachelier-eta-q}. Substitution into \eqref{eq:M-operator} gives \eqref{eq:bachelier-M-L}. Finally,
\begin{equation*}
  C''(K)
  =\frac{dh}{dK}
  =h'(y)\frac{dy}{dK}
  =\phi(y)m(y)\frac{q(y)^2}{q(y)-yq'(y)},
\end{equation*}
and \eqref{eq:bachelier-density-q} follows.
\end{proof}

\begin{theorem}[Bachelier linear butterfly criterion]
\label{thm:bachelier-linear}
Assume $q(y)>0$ and $q(y)-yq'(y)>0$. Then
\begin{equation}
\label{eq:bachelier-linear-condition}
  C''(K)\ge0
  \quad\Longleftrightarrow\quad
  q''(y)+yq'(y)-q(y)\le0.
\end{equation}
\end{theorem}

\begin{proof}
In \eqref{eq:bachelier-density-q}, all multiplicative factors outside the final parenthesis have fixed positive sign after the leading minus sign. The two inequalities are therefore equivalent.
\end{proof}

Define the linear operator
\begin{equation}
\label{eq:L-operator}
  \mathcal L q=q''+yq'-q.
\end{equation}
Then $m\ge0$ is equivalent to $\mathcal Lq\le0$. The linear Bachelier condition is therefore not a separate phenomenon: it is exactly the general Gaussian density-deformation condition expressed through reciprocal implied scale.

\subsection{Equality solutions}

\begin{proposition}[Fundamental solutions of the Bachelier equality equation]
\label{prop:bachelier-fundamental}
The functions $c$ and $p$ in \eqref{eq:bachelier-kernels} satisfy
\begin{equation*}
  \mathcal Lc=0,
  \qquad
  \mathcal Lp=0,
\end{equation*}
and their Wronskian is
\begin{equation}
\label{eq:wronskian}
  W(c,p)=c(y)p'(y)-c'(y)p(y)=\phi(y)>0.
\end{equation}
Hence every solution of $\mathcal Lq=0$ on an interval is
\begin{equation}
\label{eq:q-equality}
  q(y)=A c(y)+B p(y)
\end{equation}
for constants $A,B$.
\end{proposition}

\begin{proof}
Direct differentiation gives
\begin{equation*}
  c'(y)=-\Phi(-y),
  \qquad
  c''(y)=\phi(y),
\end{equation*}
so
\begin{equation*}
  c''(y)+yc'(y)-c(y)=0.
\end{equation*}
Similarly,
\begin{equation*}
  p'(y)=\Phi(y),
  \qquad
  p''(y)=\phi(y),
\end{equation*}
so $\mathcal Lp=0$. The Wronskian equals
\begin{equation*}
  c(y)\Phi(y)+\Phi(-y)p(y)
  =\phi(y)\bigl(\Phi(y)+\Phi(-y)\bigr)
  =\phi(y).
\end{equation*}
Thus $c$ and $p$ form a fundamental pair.
\end{proof}

Because $\mathcal Lq=0$ is equivalent to $m=0$, these equality profiles are precisely the reciprocal-scale representations of zero-density intervals. This ties the ODE equality theory back to the central object $\eta$: along such intervals the CDF is locally constant, and $\eta$ follows one of the zero-density deformation curves of Proposition~\ref{prop:mills-zero-density}.

\begin{example}[Constant normal volatility]
If $s(y)\equiv s_0$, then $q(y)\equiv1/s_0$, so
\begin{equation*}
  \mathcal Lq=-q<0.
\end{equation*}
Also $\eta=0$ and $m=1$. Constant Bachelier volatility therefore corresponds exactly to the Gaussian benchmark and lies strictly inside the positive-density region.
\end{example}

\section{Black coordinates: the same CDF deformation in two Gaussian coordinates}
\label{sec:black}

The Black convention gives the same interpretation of $\eta$ in a different geometry. The slope of total Black implied volatility is again the Gaussian-normalized difference between an option-implied CDF and the standard normal CDF. The two Fukasawa coordinates provide two measure-dependent descriptions of that same deformation.

Let
\begin{equation*}
  k=\log(K/F)
\end{equation*}
be log-forward moneyness, and let $v(k)>0$ denote total Black implied volatility. Define
\begin{equation}
\label{eq:ab-def}
  a(k)=\frac{k}{v(k)}-\frac{v(k)}{2},
  \qquad
  b(k)=\frac{k}{v(k)}+\frac{v(k)}{2}.
\end{equation}
These are $-d_1$ and $-d_2$. Assume that $a$ and $b$ are differentiable and increasing. Their monotonicity under no-arbitrage is discussed by Fukasawa~\cite{Fukasawa2012}, Lucic~\cite{Lucic2021}, and Sun~\cite{Sun2026}.

Define
\begin{equation}
\label{eq:black-eta}
  \eta(k)=v'(k)=\frac{dv(k)}{dk}.
\end{equation}
Differentiating \eqref{eq:ab-def} gives
\begin{equation}
\label{eq:ab-derivatives}
  a'(k)=\frac{1-b(k)\eta(k)}{v(k)},
  \qquad
  b'(k)=\frac{1-a(k)\eta(k)}{v(k)}.
\end{equation}

\subsection{Risk-neutral and share-measure CDF deformations}

Normalize the forward to one. The Black call price is
\begin{equation}
\label{eq:black-call}
  c(k)=\Phi(-a(k))-e^k\Phi(-b(k)).
\end{equation}
Let $X=S_T/F$, so $\E[X]=1$. Denote by $h_Q(k)=\Q(\log X\le k)$ the pricing-measure CDF and by $h_S(k)$ the corresponding CDF under the share measure with density $X$ relative to $\Q$.

\begin{proposition}[Black CDF deformations]
\label{prop:black-cdfs}
The pricing-measure CDF satisfies
\begin{equation}
\label{eq:black-cdf-f2}
  h_Q(k)=\Phi(b(k))+\phi(b(k))\eta(k),
\end{equation}
and the share-measure CDF satisfies
\begin{equation}
\label{eq:black-cdf-f1}
  h_S(k)=\Phi(a(k))+\phi(a(k))\eta(k).
\end{equation}
Equivalently,
\begin{equation}
\label{eq:black-eta-cdf}
  \eta(k)
  =\frac{h_Q(k)-\Phi(b(k))}{\phi(b(k))}
  =\frac{h_S(k)-\Phi(a(k))}{\phi(a(k))}.
\end{equation}
\end{proposition}

\begin{proof}
Differentiating \eqref{eq:black-call} with respect to $k$ and using
\begin{equation*}
  \phi(a(k))=e^k\phi(b(k))
\end{equation*}
gives
\begin{equation*}
  c'(k)=e^k\bigl(\phi(b(k))\eta(k)-\Phi(-b(k))\bigr).
\end{equation*}
Since $K=Fe^k$, the strike derivative of the normalized call is $e^{-k}c'(k)$. Therefore
\begin{equation*}
  h_Q(k)
  =1+e^{-k}c'(k)
  =\Phi(b(k))+\phi(b(k))\eta(k).
\end{equation*}

The normalized Black put price is
\begin{equation*}
  p(k)=e^k\Phi(b(k))-\Phi(a(k)).
\end{equation*}
Also,
\begin{equation*}
  p(k)=\E\bigl((e^k-X)^+\bigr)
  =e^kh_Q(k)-h_S(k).
\end{equation*}
Substituting the first CDF identity and using $e^k\phi(b(k))=\phi(a(k))$ gives
\begin{equation*}
  h_S(k)=\Phi(a(k))+\phi(a(k))\eta(k).
\end{equation*}
\end{proof}

Thus the same implied-volatility slope has two distributional meanings. In the $b$ coordinate it measures the pricing-measure CDF deformation, while in the $a$ coordinate it measures the share-measure CDF deformation. The two normalized coordinates are therefore not merely algebraic reparameterizations of Black implied volatility; they are Gaussian lenses for two economically natural measures.

When $\eta$ is viewed as a function of $b$ through the inverse map $k=k(b)$, Theorem~\ref{thm:gaussian-skew} gives
\begin{equation*}
  m_b(b)=\mathcal M_b(\eta)
  =1+\frac{d\eta}{db}-b\eta.
\end{equation*}
Similarly, in the $a$ coordinate,
\begin{equation*}
  m_a(a)=\mathcal M_a(\eta)
  =1+\frac{d\eta}{da}-a\eta.
\end{equation*}
The Mills-ratio bounds become
\begin{equation}
\label{eq:black-mills-b}
  -R(-b)\le\eta\le R(b),
\end{equation}
and
\begin{equation}
\label{eq:black-mills-a}
  -R(-a)\le\eta\le R(a).
\end{equation}

\subsection{The common butterfly factor}

\begin{theorem}[Black density deformation identity]
\label{thm:black-density}
Under the increasing-coordinate assumptions,
\begin{equation}
\label{eq:black-M-density}
  b'(k)m_b(b(k))
  =a'(k)m_a(a(k))
  =v''(k)+v(k)a'(k)b'(k).
\end{equation}
Consequently, risk-neutral density positivity is equivalent to any of
\begin{equation}
\label{eq:black-equivalent-density}
  m_b\ge0,
  \qquad
  m_a\ge0,
  \qquad
  v''(k)+v(k)a'(k)b'(k)\ge0.
\end{equation}
\end{theorem}

\begin{proof}
Since $\eta(k)=v'(k)$,
\begin{equation*}
  \frac{d\eta}{db}=\frac{v''(k)}{b'(k)}.
\end{equation*}
Therefore
\begin{align*}
  b'(k)m_b
  &=b'(k)+v''(k)-b(k)\eta(k)b'(k)\\
  &=v''(k)+b'(k)\bigl(1-b(k)\eta(k)\bigr)\\
  &=v''(k)+v(k)a'(k)b'(k),
\end{align*}
where the last equality uses \eqref{eq:ab-derivatives}. The calculation in the $a$ coordinate is identical.
\end{proof}

The familiar Black butterfly factor is therefore not an unrelated nonlinear expression. It is the Gaussian-relative density deformation generated by $\eta$, translated back from either normalized coordinate to log-forward moneyness. The first-order monotonicity of $a$ and $b$ fixes orientation; $\eta$ describes the CDF displacement; and $m_a,m_b$ describe the corresponding density displacement.

\begin{example}[Constant Black volatility]
If $v(k)\equiv v_0$, then $\eta=0$ and
\begin{equation*}
  a'(k)=b'(k)=\frac{1}{v_0}.
\end{equation*}
Hence $m_a=m_b=1$ and
\begin{equation*}
  v''(k)+v(k)a'(k)b'(k)=\frac{1}{v_0}>0.
\end{equation*}
Constant Black volatility therefore corresponds exactly to the Gaussian benchmark in both normalized coordinates.
\end{example}

\section{One-coordinate reconstruction in the Black convention}
\label{sec:reconstruction}

The deformation $\eta$ can be treated as the primitive object rather than as a quantity computed only after a smile has been specified. This viewpoint separates two tasks: the Mills and density conditions screen a candidate $\eta$ for no-arbitrage, while a first-order equation reconstructs the implied-volatility smile.

Take $b$ as the independent variable and write
\begin{equation}
\label{eq:a-of-b}
  a=A(b).
\end{equation}
Then
\begin{equation}
\label{eq:vk-from-ab}
  v(b)=b-A(b),
  \qquad
  k(b)=\frac{b^2-A(b)^2}{2}.
\end{equation}

\begin{proposition}[One-coordinate deformation formula]
\label{prop:one-coordinate}
If $b-A(b)A'(b)>0$, then
\begin{equation}
\label{eq:eta-A}
  \eta(b)
  =\frac{dv/d b}{dk/d b}
  =\frac{1-A'(b)}{b-A(b)A'(b)}.
\end{equation}
Conversely, if $v=v(b)>0$ and $a=b-v$, then $v$ satisfies
\begin{equation}
\label{eq:black-reconstruction-ode}
  v'(b)
  =\frac{\eta(b)v(b)}
  {1-(b-v(b))\eta(b)}.
\end{equation}
\end{proposition}

\begin{proof}
Differentiating \eqref{eq:vk-from-ab} gives
\begin{equation*}
  v'(b)=1-A'(b),
  \qquad
  k'(b)=b-A(b)A'(b),
\end{equation*}
which proves \eqref{eq:eta-A}. Alternatively,
\begin{equation*}
  k(b)=b v(b)-\frac{v(b)^2}{2},
\end{equation*}
so
\begin{equation*}
  k'(b)=v(b)+(b-v(b))v'(b).
\end{equation*}
Using $\eta=(dv/db)/(dk/db)$ and solving for $v'(b)$ gives \eqref{eq:black-reconstruction-ode}.
\end{proof}

A candidate deformation $\eta(b)$ can therefore be screened by the CDF barriers
\begin{equation*}
  -R(-b)\le\eta(b)\le R(b),
\end{equation*}
the coordinate orientation
\begin{equation*}
  1-(b-v(b))\eta(b)>0,
\end{equation*}
and the density condition
\begin{equation*}
  \mathcal M_b(\eta)\ge0.
\end{equation*}
The reconstruction equation then recovers $v$. This makes $\eta$ a natural candidate parameterization variable: it directly controls CDF displacement, its differential transform controls density positivity, and only afterward is the implied-volatility level reconstructed.

\section{Surface extension and local volatility}
\label{sec:surface}

The deformation viewpoint also clarifies the surface problem. Static no-arbitrage is encoded in the spatial density deformation $m$, while the maturity derivative describes how the implied scale evolves across expiries. Local variance is obtained by dividing the maturity evolution by the same second-order factor that enforces density positivity.

To keep the formulas transparent, this section works in forward units under zero rates and carry. Deterministic carry can be handled by the usual forward normalization.

\subsection{Bachelier local variance}

Let $s=s(y,t)$ denote total normal implied standard deviation at maturity $t$ and set
\begin{equation*}
  q(y,t)=\frac{1}{s(y,t)}.
\end{equation*}
At each maturity,
\begin{equation*}
  K=F+\frac{y}{q(y,t)}.
\end{equation*}
Define
\begin{equation*}
  \mathcal Lq=q_{yy}+yq_y-q.
\end{equation*}

\begin{proposition}[Bachelier local-variance formula]
\label{prop:bachelier-local}
Assume $q-yq_y>0$ and $\mathcal Lq<0$. Then
\begin{equation}
\label{eq:bachelier-local}
  a_{\mathrm{loc}}^2
  =2\frac{q_t}{\mathcal Lq}
  \frac{(q-yq_y)^2}{q^4}.
\end{equation}
\end{proposition}

\begin{proof}
At fixed maturity, Proposition~\ref{prop:bachelier-M} gives
\begin{equation*}
  \frac{\partial^2 C}{\partial K^2}
  =-\frac{\phi(y)q^3}{(q-yq_y)^3}\mathcal Lq.
\end{equation*}
At fixed strike, differentiating $K-F=y/q(y,t)$ gives
\begin{equation*}
  y_t=\frac{yq_t}{q-yq_y}.
\end{equation*}
A direct differentiation of $C=q^{-1}c(y)$ at fixed strike then yields
\begin{equation*}
  \frac{\partial C}{\partial t}
  =-\frac{\phi(y)q_t}{q(q-yq_y)}.
\end{equation*}
Using the Bachelier Dupire equation
\begin{equation*}
  \frac{\partial C}{\partial t}
  =\frac12 a_{\mathrm{loc}}^2\frac{\partial^2 C}{\partial K^2}
\end{equation*}
gives \eqref{eq:bachelier-local}.
\end{proof}

\begin{corollary}[Bachelier calendar monotonicity]
Under $q>0$ and $q-yq_y>0$, calendar monotonicity $\partial C/\partial t\ge0$ at fixed strike is equivalent to
\begin{equation*}
  q_t\le0.
\end{equation*}
\end{corollary}

The denominator in \eqref{eq:bachelier-local} is the same object that determines $m(y)$. Thus Bachelier local variance compares maturity evolution with the Gaussian-relative density deformation already identified in the fixed-maturity theory.

\subsection{Black local variance}

Let $v=v(k,t)$ be total Black implied volatility. At each fixed maturity define $a$, $b$, and $\eta=\partial v/\partial k$ as in Section~\ref{sec:black}. In this section, derivatives such as $v_t$ and $v_{kk}$ are genuine partial derivatives of the two-variable surface.

\begin{proposition}[Black local-variance formula]
\label{prop:black-local}
Assume
\begin{equation*}
  v_{kk}+v a_kb_k>0.
\end{equation*}
Then
\begin{equation}
\label{eq:black-local}
  \sigma_{\mathrm{loc}}^2
  =\frac{2v_t}{v_{kk}+v a_kb_k}
  =\frac{2v_t}{b_k m_b}
  =\frac{2v_t}{a_k m_a}.
\end{equation}
\end{proposition}

\begin{proof}
At fixed $k$, Black vega with respect to total standard deviation is
\begin{equation*}
  c_v=\phi(a),
\end{equation*}
so
\begin{equation*}
  c_t=\phi(a)v_t.
\end{equation*}
The standard log-strike differentiation gives
\begin{equation*}
  c_{kk}-c_k
  =\phi(a)\bigl(v_{kk}+v a_kb_k\bigr).
\end{equation*}
Using
\begin{equation*}
  c_t
  =\frac12\sigma_{\mathrm{loc}}^2(c_{kk}-c_k)
\end{equation*}
gives the first equality in \eqref{eq:black-local}. The remaining equalities follow from the fixed-maturity identity \eqref{eq:black-M-density} applied slice by slice.
\end{proof}

\begin{corollary}[Black calendar monotonicity]
In the zero-carry forward normalization, calendar monotonicity at fixed $k$ is equivalent to
\begin{equation*}
  v_t\ge0.
\end{equation*}
\end{corollary}

Again, the same quantity has both a static and a dynamic role. The deformation $m$ is the Gaussian-relative density in a normalized coordinate, and the local-volatility denominator is that same density deformation translated back to log-forward moneyness.

\section{A Black--Bachelier deformation dictionary}
\label{sec:dictionary}

The common structure can be summarized as follows.

\begin{center}
\begin{tabular}{@{}lll@{}}
\toprule
 & Bachelier & Black \\
\midrule
Normalized coordinate
& $y=(K-F)/s$
& $a=k/v-v/2$, $b=k/v+v/2$ \\
CDF deformation
& $\eta=s'(K(y))$
& $\eta=v'(k)$ \\
CDF representation
& $h=\Phi(y)+\phi(y)\eta$
& $h_Q=\Phi(b)+\phi(b)\eta$ \\
Density deformation
& $m=\mathcal M_y(\eta)$
& $m_b=\mathcal M_b(\eta)$ \\
Second-order factor
& $-q\mathcal Lq/(q-yq')^2$
& $(v''+v a'b')/b'$ \\
Butterfly condition
& $\mathcal Lq\le0$
& $v''+v a'b'\ge0$ \\
\bottomrule
\end{tabular}
\end{center}

The unifying object is $\eta$, not the particular implied-volatility convention. In both models, implied-volatility slope becomes the Gaussian-normalized difference between an economically meaningful option-implied CDF and the standard normal CDF. Its differential transform $m$ then measures the corresponding density distortion.

The conventions differ in what happens after this common deformation has been identified. In Bachelier coordinates, reciprocal implied scale converts $m\ge0$ into the linear differential inequality $\mathcal Lq\le0$. In Black coordinates, the same density deformation becomes the familiar nonlinear butterfly factor, with the $a$ and $b$ coordinates furnishing two equivalent Gaussian representations under two natural measures.

\section{Conclusion}
\label{sec:conclusion}

The central object of this paper is the Gaussian CDF deformation
\begin{equation*}
  \eta(z)=\frac{h(z)-\Phi(z)}{\phi(z)}.
\end{equation*}
It gives implied-volatility skew a direct probabilistic meaning: $\eta$ measures the pointwise departure of the risk-neutral CDF from a Gaussian benchmark after the appropriate normalization. Its differential transform
\begin{equation*}
  m(z)=1+\eta'(z)-z\eta(z)
  =\frac{h'(z)}{\phi(z)}
\end{equation*}
measures the corresponding departure of the density. Thus
\begin{equation*}
  \eta
  \quad\longleftrightarrow\quad
  \text{CDF deformation},
  \qquad
  m
  \quad\longleftrightarrow\quad
  \text{density deformation}.
\end{equation*}
The Gaussian benchmark is simply $\eta=0$ and $m=1$.

This viewpoint turns the no-arbitrage inequality into a complete distributional characterization. The condition $m\ge0$ is equivalent to monotonicity of the option-implied CDF, and under the CDF boundary conditions $m$ is any nonnegative Gaussian-relative density of unit Gaussian-weighted mass. The Mills-ratio barriers are exactly the bounds required to keep the CDF between zero and one.

In the Bachelier convention, the deformation is
\begin{equation*}
  \eta(y)=s'(K(y)),
\end{equation*}
so total normal implied-volatility skew is itself the normalized CDF discrepancy. Reciprocal implied scale then produces the additional linearization
\begin{equation*}
  C''(K)\ge0
  \quad\Longleftrightarrow\quad
  q''(y)+yq'(y)-q(y)\le0.
\end{equation*}
In the Black convention,
\begin{equation*}
  \eta=v'(k)
\end{equation*}
plays the same role in both Fukasawa coordinates, under the pricing and share measures respectively, and its density deformation reproduces the standard Black butterfly factor.

The local-volatility formulas show that this interpretation survives at the surface level: the same density deformation that enforces static butterfly no-arbitrage also appears in the denominator of local variance. The normalized-coordinate framework therefore links implied-volatility slope, CDF displacement, density positivity, and local-volatility inversion through one object, $\eta$.

A separate companion paper develops a model-free projective option-price coordinate system. The present paper remains within Gaussian implied-volatility conventions and isolates the distributional meaning of their normalized slopes.

\end{document}